\documentclass{amsart}

\usepackage{amssymb}
\usepackage{amsthm}
\usepackage{mathrsfs}
\usepackage{amsmath}
\usepackage{graphicx}
\usepackage{array}

\newtheorem{example}{Example}
\newtheorem{definition}{Definition}

\newtheorem{theorem}{Theorem}
\newtheorem{proposition}{Proposition}
\newtheorem{remark}{Remark}
\newtheorem{lemma}{Lemma}

\begin{document}

\title{Complex exceptional orthogonal polynomials and quasi-invariance}
\author[W.~A.~Haese-Hill]{William A.~Haese-Hill}
\address{Department of Mathematical Sciences\\ 
Loughborough University, Loughborough LE11 3TU, UK}
\author[M.~A.~Halln\"as]{Martin A.~Halln\"as}
\address{Department of Mathematical Sciences\\ 
Loughborough University, Loughborough LE11 3TU, UK}
\author[A.~P.~Veselov]{Alexander P.~Veselov}
\address{Department of Mathematical Sciences\\ 
Loughborough University, Loughborough LE11 3TU,  UK\\
and Moscow State University, Russia}

\date{\today}

\begin{abstract}
Consider the Wronskians of the classical Hermite polynomials 
$$H_{\lambda, l}(x):=\mathrm{Wr}(H_l(x),H_{k_1}(x)\ldots, H_{k_n}(x)), \quad l \in \mathbb Z_{\geq 0},$$
where $k_i=\lambda_i+n-i, \,\, i=1,\dots, n$ and $\lambda=(\lambda_1, \dots, \lambda_n)$ is a partition. 
G\'omez-Ullate et al showed that for a special class of partitions the corresponding polynomials are orthogonal and dense among all polynomials with certain inner product, but in contrast to the usual case have some degrees missing (so called exceptional orthogonal polynomials). We generalise their results to all partitions by considering complex contours of integration and non-positive Hermitian products. The corresponding polynomials are orthogonal and dense in a finite-codimensional subspace of $\mathbb C[x]$ satisfying certain quasi-invariance conditions. 
A Laurent version of exceptional orthogonal polynomials, related to monodromy-free trigonometric 
Schr\"odinger operators,  is also presented. 
\end{abstract}

\maketitle

\section{Introduction}
Consider polynomials $p_n(x) \in \mathbb R[x]$ of degrees $n=0,1,\dots$, satisfying the orthogonality relation
$$
	(p_m,p_n) =\delta_{mn}g_n,
$$
where the inner product of polynomials is defined by a real integral
\begin{equation}
\label{real}
(p,q):= \int_a^b p(x) q(x) w(x) dx
\end{equation}
for some positive weight function $w.$
Suppose that there exists a second order differential operator 
$$
T= A(x)\frac{d^2}{dx^2}+B(x)\frac{d}{dx}+C(x)
$$
having these polynomials as eigenvectors:
$$
T p_n(x)=E_n p_n(x), \quad n=0,1,\dots.
$$
A classical result due to  Bochner \cite{Bochner} says that in that case the sequence of polynomials $p_n(x), \, n \in \mathbb Z_{\geq 0}$, must coincide (up to a linear change of $x$) with one of the systems of classical orthogonal polynomials of Hermite, Laguerre or Jacobi.

G\'omez-Ullate, Kamran and Milson \cite{UKM-boch} considered the following variation of Bochner's question.
Let us assume now that in the previous considerations $n$ belongs to a certain proper subset $S \subset \mathbb Z_{\geq 0}$ such that $\mathbb Z_{\geq 0} \setminus S$ is finite. To make this non-trivial they added the following {\it density condition}: the linear span $U=\langle p_n:  n \in S \rangle $ of the corresponding polynomials must be dense in $\mathbb R[x]$ in the sense that if $(p, p_n)=0$ for all $n \in S$ then $p\equiv 0.$ In that case the sequence $p_n(x), \, n \in S$ is called a system of {\it exceptional orthogonal polynomials.}

The main example of such polynomials are exceptional Hermite polynomials \cite{UGM-EHP} having the Wronskian form
\begin{equation}
\label{excH}
H_{\lambda, l}(x):=\mathrm{Wr}(H_l(x),H_{k_1}(x)\ldots, H_{k_n}(x)), \quad l \in \mathbb Z_{\geq 0},
\end{equation}
where $H_l(x)$ are classical Hermite polynomials, $\lambda=(\lambda_1, \dots, \lambda_n)$ is a double partition and 
$$
k_i=\lambda_i+n-i, \,\, i=1,\dots, n.
$$
The double partitions have the very special form $$\lambda=\mu^2=(\mu_1,\mu_1, \mu_2, \mu_2, \dots, \mu_k, \mu_k),$$ where $\mu=(\mu_1,\mu_2, \dots, \mu_k)$ is another partition with $n=2k$ (see \cite{FHV-zeroes}). According to Krein and Adler \cite{Adler} this guarantees that the corresponding Wronskian
\begin{equation}
\label{W}
W_\lambda(x) = \mathrm{Wr}(H_{k_1}(x)\ldots, H_{k_n}(x))
\end{equation}
 has no zeroes on the real line and thus determines a non-singular weight function
\begin{equation}\label{w}
w(x)=W_\lambda^{-2}(x)e^{-x^2}.
\end{equation}
The geometry of the complex zeroes of the corresponding Wronskians is quite interesting and was studied by Felder et al.~in \cite{FHV-zeroes}.

One of the goals of our paper is to find a proper interpretation of the exceptional Hermite polynomials (\ref{excH}) for all partitions $\lambda.$
As we will see, this will naturally lead us to the notion of {\it quasi-invariance}, which appeared in the theory of {\it monodromy-free Schr\"odinger operators}, going back to Picard and Darboux and more recently revisited by Duistermaat and Gr\"unbaum \cite{DG-spec}. In certain classes such operators were explicitly described in terms of Wronskians in \cite{DG-spec,CFV-Huygen, Obl, GV}. 
Grinevich and Novikov  studied the spectral properties of these and more general singular finite-gap operators and emphasized the important link with the theory of Pontrjagin spaces (see \cite{GN} and references therein).
Our paper can be considered as dealing with the implications of all these results for the theory of exceptional orthogonal polynomials. 

More precisely, we first complexify the picture by considering the vector space $V=\mathbb C[z]$ and replace the inner product (\ref{real}) by a Hermitian product
of the form 
\begin{equation}
\label{compl}
\langle p,q \rangle:= \int_C p(z) {\bar q}(z) w(z) dz,
\end{equation}
where ${\bar q}(z):=\overline{q(\bar z)}$ is the Schwarz conjugate of the polynomial $q(z)$, $C \subset \mathbb C$ is a contour in the complex domain and $w(z)$ is a complex weight function. The condition that this product is Hermitian 
implies certain restrictions on the contour $C$ and function $w(z)$ (see section 2). It also requires certain restrictions on the set of polynomials for which the product is well defined. As it turned out, such polynomials form a subspace $U \subset V$ of finite codimension defined by some quasi-invariance conditions. 
Similarly to \cite{UKM-boch} we say that the polynomials $p_n(z), \, n \in S$, form a system of {\it complex exceptional orthogonal polynomials} if their linear span is a subspace of $U$ that is dense in $U$ in the sense that $\langle p, p_n \rangle=0$ for all $n \in S$, implies that $p \equiv 0.$

We will show that the Wronskians (\ref{excH})  satisfy this criteria for every partition $\lambda$ and a suitable choice of $C$ with $w$ given by \eqref{w}. For a double partition $\lambda$ we can take as a contour $C$ the real line with $U=V$ and recover the results by G\'omez-Ullate et al \cite{UGM-EHP}. 

Note that the corresponding Hermitian form is positive definite only for double partitions, otherwise we always have polynomials with negative norms. The appearance of negative norms for singular potentials was first emphasized by Grinevich and Novikov \cite{GN}.

We also consider the Laurent version of our approach. Some Laurent versions of orthogonal polynomials are already known in the literature (see e.g.~\cite{OLP} and references therein), but our approach is different since it is not based on Gram-Schmidt procedure. Similarly, it does not fit into the theory of orthogonal polynomials on the unit circle initiated by Szeg\"o \cite{Szego}, who considered the case of usual polynomials.

Consider the Laurent polynomials $\Lambda=\mathbb C [z, z^{-1}]$ and the following complex bilinear form on $\Lambda$:
$$
(P,Q)=\frac{1}{2\pi i}\oint_{C} P(z)Q(z) \frac{\, dz}{z}
$$
where $C=\{z \in \mathbb C: |z|=1\}$ is the unit circle. The standard basis $z^n, \, n \in \mathbb Z$, satisfies the {\it Laurent orthogonality relation}
$$
(z^k, z^l)=\delta_{k+l,0}, \quad k,l \in \mathbb Z.
$$

We consider more general forms
\begin{equation}
\label{complex_euclid}
\left(P,Q\right)=\frac{1}{2\pi i}\oint_{C_\mu} P(z)Q(z) w(z) \frac{\, dz}{z},
\end{equation}
where $C_\mu$ is the circle defined by $|z|=\mu$ and $w(z)=W(z)^{-2}$, with $W(z)$ some Laurent polynomial. For this form to be well-defined, we need to assume that $P, Q$ belong to a suitable subspace of quasi-invariants $\mathscr{Q} \subset \Lambda$ of finite codimension.



Let $\mathcal K$ be a finite subset of $\mathbb N$. Suppose that $P_n \in \Lambda, \, n \in \mathbb Z$,  satisfy the Laurent orthogonality relation
\begin{equation}
\label{lort}
\left(P_k, P_l\right)=\delta_{k+l,0} h_k, \quad k,l \in \mathbb Z,
\end{equation}
but $P_n$ is proportional to $P_{-n}$ for $n\in\mathcal K$, which implies that the corresponding $h_n=0$, and thus $P_n$ is orthogonal to all $P_k, \, k \in \mathbb Z.$ If the minimal complex Euclidean extension of the linear span of $P_n, \, n \in \mathbb Z$, coincides with the subspace of quasi-invariants $\mathscr{Q}$, then we call them {\it exceptional Laurent orthogonal polynomials}. The need to consider such an extension is the novelty of the Laurent case, which is related to the fact that the corresponding form is degenerate on the linear span of $P_n,\, n \in \mathbb Z$.

We present an example of such polynomials corresponding to the trigonometric monodromy-free Schr\"odinger operators \cite{CFV-Huygen}. Namely, for any set $\kappa=\{k_1, \dots, k_n\}$ of distinct natural numbers $k_1>k_2>\dots >k_n>0$ and any choice of complex parameters $a=(a_1,\dots, a_n), \, a_k \in \mathbb C \setminus \{0\}$, we define the Laurent polynomials
\begin{equation}
\label{excL}
	P_{\kappa,a; l}(z) =
 		\begin{vmatrix}
  			\Phi_{k_1}(a_1;z)			&	\Phi_{k_2}	(a_2;z)		&	\cdots 	&	\Phi_{k_n}(a_n;z) 			&	z^l			\\
 			D \Phi_{k_1}(a_1;z)		&	D \Phi_{k_2}(a_2;z) 		&	\cdots 	&	D \Phi_{k_n}(a_n;z) 		&	D z^{l}		\\
 			\vdots  			& 	\vdots 			& 	\ddots 	& 	\vdots  			&	\vdots			\\
 			D^{n}\Phi_{k_1}(a_1;z)	&	D^{n} \Phi_{k_2}(a_2;z) 	&	\cdots 	&	D^{n} \Phi_{k_n}(a_n;z) 	&	D^{n} z^{l}
 		\end{vmatrix}
\end{equation}
where $\Phi_k(a;z)=a z^k+a^{-1} z^{-k}, \, k \in \mathbb N$ and $D=z\frac{d}{dz}.$ 

When parameters $a_k$ satisfy the condition $|a_k|=1$ for all $k=1,\dots n$ we introduce a Hermitian form on a certain subspace of quasi-invariant Laurent polynomials $\mathscr{Q}_{\kappa,C}$ and show that the minimal Hermitian extension of the linear span of $P_{\kappa,a; l}$, $l\in\mathbb Z$, coincides with the subspace of quasi-invariants $\mathscr{Q}_{\kappa}$ and is dense in $\mathscr{Q}_{\kappa,C}$.

\section{Complex exceptional Hermite polynomials}\label{Sec2}

In this section we consider the polynomials $H_{\lambda,l}$, as defined in \eqref{excH}, for general partitions, i.e., we do not require that $\lambda$ is a double partition. We shall refer to these polynomials as {\it complex exceptional Hermite polynomial} or CEHPs for short. 
Although these polynomials have real coefficients for the partitions which are not double it is natural to consider them as elements of the complex Hermitian vector space because the contour of integration in (\ref{compl}) is complex.


We begin by recalling how $H_{\lambda,l}$ are obtained by a sequence of Darboux transformations from the classical Hermite polynomial $H_l$. The starting point is the classical fact that the functions
\begin{equation}
\label{psil}
\psi_l(z)=H_l(z)e^{-z^2/2},\ \ \ l\in\mathbb{Z}_{\geq 0}
\end{equation}
have the eigenfunction property
$$
\mathscr{L}\psi_l\equiv -\frac{d^2\psi_l}{dz^2}+z^2\psi=(2l+1)\psi_l,\ \ \ z\in\mathbb{C},
$$
and satisfy the boundary conditions
$$
\lim_{\mathrm{Re}\, z\to\pm\infty}\psi_l(z)=0
$$
(for any fixed value of $\mathrm{Im}\, z$). We will choose the normalisation of Hermite polynomials such that the highest coefficient of $H_l(z)$ is $2^l$ and all the coefficients are integer:
$$H_0=1, \, H_1=2z, \, H_2=4z^2-2, \, H_3=8z^3-12z, \, H_4=16z^4-48z^2+12,\dots .$$
%
%

As is well known, after $n$ consecutive Darboux transformations at the levels $k_n<k_{n-1}<\cdots<k_1$, where $k_i=\lambda_i+n-i, \, i=1,\dots,n,$ one arrives at the Schr\"odinger operator
\begin{equation}
\label{LlamH}
\mathscr{L}_\lambda=-\frac{d^2}{dz^2}-2\frac{d^2}{dz^2}\big(\log {\rm Wr}(\psi_{k_1},\ldots,\psi_{k_n})\big)+z^2,
\end{equation}
which satisfies the intertwining relation
$$
\mathscr{D}_\lambda\circ \mathscr{L}=\mathscr{L}_\lambda\circ \mathscr{D}_\lambda,
$$
where the intertwining operator $\mathscr{D}_\lambda$ acts according to
\begin{equation}\label{inttw}
\mathscr{D}_\lambda\psi=\frac{{\rm Wr}(\psi,\psi_{k_1},\ldots,\psi_{k_n})}{{\rm Wr}(\psi_{k_1},\ldots,\psi_{k_n})},
\end{equation}
see e.g.~\cite{Adler,Crum}. It follows that the functions
\begin{equation}
\label{psilaml}
\psi_{\lambda,l}=\frac{{\rm Wr}(\psi_l,\psi_{k_1},\ldots,\psi_{k_n})}{{\rm Wr}(\psi_{k_1},\ldots,\psi_{k_n})},\ \ \ l\notin{k_1,\ldots,k_n},
\end{equation}
have the eigenfunction property
$$
\mathscr{L}_\lambda\psi_{\lambda,l}=(2l+1)\psi_{\lambda,l}.
$$

Substituting \eqref{psil} in \eqref{psilaml} and using the general property ${\rm Wr}(gf_1,\ldots,gf_n)=g^n{\rm Wr}(f_1,\ldots,f_n)$, one finds that
\begin{equation}
\label{psiH}
\psi_{\lambda,l}=H_{\lambda,l}\frac{e^{-x^2/2}}{W_\lambda},
\end{equation}
where $H_{\lambda,l}$ are given by \eqref{excH}.

By a direct computation, it is readily inferred that $H_{\lambda,l}$ is an eigenfunction of the operator
$$
T_\lambda=-\frac{d^2}{dz^2}+2\left(z+\frac{W_\lambda^\prime}{W_\lambda}\right)\frac{d}{dz}-\left(\frac{W_\lambda^{\prime\prime}}{W_\lambda}+2z\frac{W_\lambda^\prime}{W_\lambda}\right)
$$
with eigenvalue $2l+1$.

\begin{example}
Consider the special case $\lambda=(1)$, which corresponds to the Schr\"odinger operator
\begin{align*}
\mathscr{L}_{(1)}&=-\frac{d^2}{dz^2}-2\frac{d^2}{dz^2}\big(\log(2ze^{-z^2/2})\big)+z^2\\
&=-\frac{d^2}{dz^2}+z^2+\frac{2}{z^2}+2.
\end{align*}
Already in this simple example, we obtain eigenfunctions \eqref{psilaml} with a singularity on the real line (at $z=0$). Indeed, this can be seen explicitly by writing out the first exceptional Hermite polynomials $H_{(1),k}=\mathrm{Wr}(\psi_k,\psi_1)$ and the corresponding few eigenfunctions $\psi_{(1),k}=\frac{\mathrm{Wr}(\psi_k,\psi_1)}{\psi_1}:$
\renewcommand\arraystretch{2}
$$
\begin{array}{l r}
H_{(1),0}=1, \quad \quad &\psi_{(1),0}=\frac{1}{z}e^{-z^2/2}, \\
 H_{(1),2}=-(2+4z^2), \quad \quad &\psi_{(1),2}=-\frac{2+4z^2}{z}e^{-z^2/2},\\
H_{(1),3}=-16z^3, \quad \quad &\psi_{(1),3}=-16z^2e^{-z^2/2},\\
H_{(1),4}=12(1+4z^2-4z^4), \quad \quad &\psi_{(1),4}=\frac{12(1+4z^2-4z^4)}{z}e^{-z^2/2},\\
H_{(1),5}=64z^3(5-2z^2), \quad \quad &\psi_{(1),5}=64z^2(5-2z^2)e^{-z^2/2},\\
H_{(1),6}=40(3+18z^2-36z^4+8z^6), \quad \quad &\psi_{(1),6}=-\frac{40(3+18z^2-36z^4+8z^6)}{z}e^{-z^2/2}.
\end{array}
$$
More generally, using the fact that
$$
\mathrm{Wr}(\psi_{2l+1},\psi_1)(-z)=-\mathrm{Wr}(\psi_{2l+1},\psi_1)(z),\ \ \ l\in\mathbb{N},
$$
as well as the fact that each classical Hermite polynomial $H_{2l}(z)$, $l\in\mathbb{Z}_{\geq 0}$, has a nonzero constant term, it is readily seen that
$\psi_{(1),l}(x)$ is regular on the whole real line if and only if $l\in\mathbb{Z}_{\geq 0}\setminus\{1\}$ is odd. The eigenvalues of the first few eigenfunctions are given in Figure \ref{Fig:spect}, where open and filled circles indicate that the corresponding eigenfunctions are singular and non-singular, respectively. In addition, the cross represents the eigenvalue removed by the Darboux transformation.

Note that in the theory of quantum Calogero-Moser systems (of which this example is the simplest case) only non-singular solutions are considered (see e.g. \cite{OP}).

\begin{figure}
\label{Fig:spect}
\begin{center}
\includegraphics[width=1.0\textwidth]{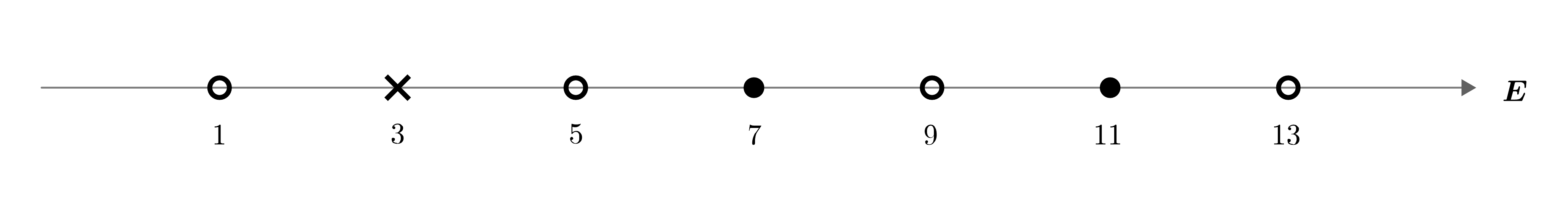}
\end{center}
\caption{The eigenvalues of the first few eigenfunctions for $\lambda=(1)$.}
\end{figure}
%
\end{example}

We will now use the fact that $\mathscr{D}_\lambda$ is obtained as the composition of first order intertwining operators. To be more specific, let us introduce the short hand notation
$$
\mathscr{W}_m={\rm Wr}(\psi_{k_m},\ldots,\psi_{k_n}),\ \ \ \mathscr{W}_m(\psi)={\rm Wr}(\psi,\psi_{k_m},\ldots,\psi_{k_n}),
$$
(where it is convenient to allow $m=n+1$ and set $\mathscr{W}_{n+1}=1$, $\mathscr{W}_{n+1}(\psi)=\psi$), and recall the standard identity
$$
\mathscr{W}_{m-1}\mathscr{W}_m(\psi)=\mathscr{W}_m\frac{d}{dx}\mathscr{W}_{m-1}(\psi)-\mathscr{W}_{m-1}(\psi)\frac{d}{dx}\mathscr{W}_m,\ \ \ m\geq 1.
$$
Then it is readily verified that
\begin{equation}
\label{Dfact}
\mathscr{D}_\lambda=D_1\circ\cdots\circ D_m\circ\cdots\circ D_n,
\end{equation}
with
\begin{equation}
\label{Dm}
D_m=\frac{d}{dz}-\frac{d}{dz}\left(\log\frac{\mathscr{W}_m}{\mathscr{W}_{m+1}}\right).
\end{equation}

For our purposes, a key notion is that of trivial monodromy, see e.g.~\cite{V-Stieltjes}. A Schr\"odinger operator $\mathscr{L}=-d^2/dz^2+u(z)$, whose potential $u$ is a meromorphic function of $z$, is said to have \emph{trivial monodromy} if all solutions of its eigenvalue equation
\begin{equation}
\label{LEq}
\mathscr{L}\psi(z)=E\psi(z)
\end{equation}
are meromorphic in $z$ for all $E$.

We recall that every monodromy-free Schr\"odinger operator $\mathscr{L}$ with a quadratically increasing rational potential is of the form \eqref{LlamH} for some partition $\lambda$. The fact that each Schr\"odinger operator $\mathscr{L}_\lambda$ has trivial monodromy is easily seen. Indeed, in the special case $u(z)=z^2$ all eigenfunctions are entire, and trivial monodromy is preserved under (rational) Darboux transformations. The converse result is due to Oblomkov \cite{Obl}.

Duistermaat and Gr\"unbaum \cite{DG-spec} obtained local conditions for trivial monodromy. Specifically, in a neighbourhood of a pole $z=z_i$ the potential $u(z)$ must have a Laurent series expansion of the form
$$
u(z)=\sum_{r\geq -2} c_r(z-z_i)^r,
$$
with 
$$
c_{-2}=m_i(m_i+1)\ \ \ \text{for some}\ \ \ m_i\in\mathbb{N},
$$
and
$$
c_{2j-1}=0,\ \ \ \forall j=0,1,\ldots,m_i.
$$
In addition, every eigenfunction $\psi$ has a Laurent series expansion of the form
$$
\psi(z)=(z-z_i)^{-m_i}\sum_{r=0}^\infty d_r(z-z_i)^r,
$$
with
$$
d_{2j-1}=0,\ \ \ \forall j=1,\ldots,m_i.
$$

We proceed to consider the implications for the CEHPs $H_{\lambda,l}$. Let $Z_\lambda$ be the set of zeros $z_i\in\mathbb{C}$ of the Wronskian $W_\lambda(z)$ with multiplicities $m_i\in\mathbb{N}$. In addition, we need the subset $Z_\lambda^\mathbb{R}\subset Z_\lambda$ obtained by restriction to $z_i\in\mathbb{R}$. We say that a meromorphic function $\psi(z)$ is {\it quasi-invariant} at the point $z=z_i$ with multiplicity $m_i$ if it satisfies the following two conditions:
\begin{enumerate}
\item $\psi(z)(z-z_i)^{m_i}$ is analytic at $z=z_i$,
\item $(\psi(z)(z-z_i)^{m_i})^{(2j-1)}\arrowvert_{z=z_i}=0$, for all $j=1,\ldots,m_i$.
\end{enumerate}
The second condition can be rewritten as
$$
\psi(\sigma_i(z))=(-1)^{m_i}\psi(z)+O((z-z_i)^{m_i}),
$$
where $\sigma_i(z)=2z_i-z$ is the reflection with respect to $z_i.$ This explains the terminology.

Introducing the subspace
$$
\mathcal{Q}_\lambda=\left\lbrace p\in\mathbb{C}[z] : \psi(z):=p(z)\frac{e^{-z^2/2}}{W_\lambda(z)}~\textrm{is quasi-invariant at}~z=z_i, \forall z_i\in Z_\lambda\right\rbrace,
$$
it follows from the above that the $\mathbb C$-linear span 
$$\mathcal{U}_\lambda=\langle H_{\lambda,l} : l\in\mathbb{Z}_{\geq 0}\setminus\lbrace k_1,\ldots,k_n\rbrace\rangle$$
belongs to $\mathcal{Q}_\lambda$. From Proposition 5.3 in \cite{UGM-EHP}, we recall that the codimension  of $\mathcal{U}_\lambda$ in $\mathbb{C}[z]$ is equal to $|\lambda|$. On the other hand, $|\lambda|$ is the degree of $W_\lambda(z)$, and therefore the number of quasi-invariance conditions that any $p\in\mathcal{Q}_\lambda$ should satisfy. This yields the converse inclusion, and thus the following result.

\begin{proposition}\label{Prop:Hbasis}
The $\mathbb C$-linear span of CEHPs coincides with polynomial quasi-invariants:
$$
\mathcal{U}_\lambda=\mathcal{Q}_\lambda.
$$
\end{proposition}

Whenever $\lambda$ is not a double partition, the Wronskian $W_\lambda(z)$ will have one or more real zeros \cite{Adler}, so that the weight function \eqref{w} is no longer non-singular on the real line. To resolve this problem, we replace the standard contour $\mathbb{R}$ by a shifted contour $C=i\xi+\mathbb{R}$ and consider a corresponding Hermitian product \eqref{compl}. As will become clear below, to ensure that the product is Hermitian we need to restrict attention to the following subspace of quasi-invariant polynomials:
$$
\mathcal{Q}_{\lambda,\mathbb{R}}=\left\lbrace p\in\mathbb{C}[z] : \psi(z):=p(z)\frac{e^{-z^2/2}}{W_\lambda(z)}~\textrm{is quasi-invariant at}~z=z_i, \forall z_i\in Z_\lambda^\mathbb{R}\right\rbrace.
$$

By counting quasi-invariance conditions, we obtain the next proposition.

\begin{proposition}\label{Prop:Hcodim}
The codimension of $\mathcal{Q}_\lambda$ in $\mathcal{Q}_{\lambda,\mathbb{R}}$ is $|\lambda|-\sum_{z_i\in Z_\lambda^\mathbb{R}}m_i.$
\end{proposition}

We are now ready for the main definition of this section.

\begin{definition}\label{Def:Hprod}
Let $\xi\in\mathbb{R}$ be such that
\begin{equation}
\label{xires}
0<|\xi|<|{\rm Im}~z_i|,\ \ \ \forall z_i\in Z_\lambda\setminus Z_\lambda^\mathbb{R}.
\end{equation}
Then, we define a sesquilinear product $\langle\cdot,\cdot\rangle$ on $\mathcal{Q}_{\lambda,\mathbb{R}}$ by setting
\begin{equation}
\label{Hprod}
\langle p,q\rangle=\int_{i\xi+\mathbb{R}}p(z)\bar{q}(z)\frac{e^{-z^2}}{W_\lambda^2(z)}dz,\ \ \ p,q\in\mathcal{Q}_{\lambda,\mathbb{R}}.
\end{equation}
\end{definition}

Now we will show that the product does not depend on the specific choice of $\xi$. We find it worth stressing that this important property relies on our restriction to the subspace $\mathcal{Q}_{\lambda,\mathbb{R}}$.

\begin{proposition}\label{Lemma:indep}
For any $p,q\in\mathcal{Q}_{\lambda,\mathbb{R}}$, the value of $\langle p,q\rangle$ is independent of $\xi\in\mathbb{R}$ provided the condition \eqref{xires} is satisfied.
\end{proposition}

\begin{proof}
Let $I_\xi$ denote the integral in the right-hand side of \eqref{Hprod}. By Cauchy's theorem, it suffices to show that $I_{\xi}-I_{-\xi}=0$ for some $\xi$ satisfying \eqref{xires}. From the residue theorem, we deduce that the difference between the two integrals is proportional to
$$
\sum_{z_i\in Z_{\lambda,\mathbb{R}}}\mathop{Res}_{z=z_i}\left(p(z)\bar{q}(z)\frac{e^{-z^2}}{W_\lambda^2(z)}\right).
$$
We claim that each of these residues vanish. In fact, we have the following more general result.

\begin{lemma}\label{Lemma:res}
If $\psi$, $\phi$ are quasi-invariant at $z=z_i$ with multiplicity $m_i$, then
$$
\mathop{Res}_{z=z_i}\big(\psi(z)\phi(z)\big)=0.
$$
\end{lemma}

Indeed, it follows from Condition (2) above that
\begin{multline*}
\big(\psi(z)\phi(z)(z-z_i)^{2m_i}\big)^{(2m_i-1)}\big\arrowvert_{z=z_i}\\
=\sum_{j=0}^{2m_i-1}\binom{2m_i-1}{j}\big(\psi(z)(z-z_i)^{m_i}\big)^{(2m_i-1-j)}\big\arrowvert_{z=z_i}\big(\phi(z)(z-z_i)^{m_i}\big)^{(j)}\big\arrowvert_{z=z_i}\\
=0.
\end{multline*}
\end{proof}

It is now straightforward to show that Definition \ref{Def:Hprod} yields a Hermitian product.

\begin{proposition}\label{Prop:Hprod}
The sesquilinear product $\langle\cdot,\cdot\rangle$ is Hermitian:
$$
\langle p,q\rangle=\overline{\langle q,p\rangle},\ \ \ \forall p,q\in\mathcal{Q}_{\lambda,\mathbb{R}}.
$$
\end{proposition}

\begin{proof}
In what follows, we find it convenient to use the notation
$$
w(z)=\frac{e^{-z^2/2}}{W_\lambda^2(z)},
$$
and use a subscript to indicate the choice of $\xi$ in \eqref{Hprod}. Since the classical Hermite polynomials have real coefficients, it is evident from \eqref{W} that
$\bar{w}(z)=w(z)$. Hence, we have the following equalities:
\begin{align*}
\langle p,q\rangle_\xi &=\int_\mathbb{R}p(i\xi+x)\bar{q}(i\xi+x)w(i\xi+x)dx\\
&=\overline{\int_\mathbb{R}\bar{p}(-i\xi+x)q(-i\xi+x)w(-i\xi+x)dx}\\
&=\overline{\langle q,p\rangle_{-\xi}}.
\end{align*}
Combined with Proposition \ref{Lemma:indep}, this yields the asserted hermiticity property.
\end{proof}

We recall that the classical Hermite polynomials $H_l(x)$ satisfy the orthogonality relation
\begin{equation}
\label{clHorth}
\int_\mathbb{R}H_j(x)H_l(x)e^{-x^2}dx=\delta_{jl}2^l l!\sqrt{\pi},\ \ \ j,l\in\mathbb{Z}_{\geq 0}.
\end{equation}
Combining this fact with the factorisation \eqref{Dfact} of the intertwining operator $\mathscr{D}_\lambda$, it is now readily established by induction on the length $n$ of $\lambda$ that the CEHPs $H_{\lambda,l}(x)$ are orthogonal with respect to the Hermitian form $\langle\cdot,\cdot\rangle$ (cf. \cite{UGM-EHP}).

\begin{theorem}\label{Thm:Horth}
The CEHPs $H_{\lambda,l}$ satisfy the orthogonality relation
\begin{equation}
\label{Hsqnorm}
\langle H_{\lambda,j},H_{\lambda,l}\rangle=\delta_{jl}\sqrt{\pi}2^l l!\prod_{m=1}^n 2(l-k_m),\ \ \ j,l\in\mathbb{Z}_{\geq 0}\setminus\lbrace k_1,\ldots,k_n\rbrace.
\end{equation}
\end{theorem}

\begin{proof}
The assertion clearly holds true for $n=0$, with the empty product taken to be equal to one. Introducing the partition
$$
\hat{\lambda}=(\lambda_2,\ldots,\lambda_n),
$$
we have
$$
\langle H_{\lambda,j},H_{\lambda,l}\rangle=\int_{i\xi+\mathbb{R}}(D_1\psi_{\hat{\lambda},j})(z)(\overline{D_1\psi_{\hat{\lambda},l}})(z)dz.
$$
Since $\overline{\mathscr{W}_m}=\mathscr{W}_m$, the (formal) adjoint of $D_1$ is given by
$$
D_1^*=-\frac{d}{dx}-\frac{d}{dx}\left(\log\frac{\mathscr{W}_1}{\mathscr{W}_{2}}\right).
$$
The factorisation
$$
D_1^*D_1=\mathscr{L}_{\hat{\lambda}}-2k_1-1
$$
thus entails that
$$
\langle H_{\lambda,j},H_{\lambda,l}\rangle=2(l-k_1)\langle H_{\hat{\lambda},j},H_{\hat{\lambda},l}\rangle.
$$
This completes the induction step, and the theorem is proved.
\end{proof}

\begin{remark}
Since $\langle\cdot,\cdot\rangle$ is Hermitian, each squared norm $\langle p,p\rangle$, $p\in\mathcal{Q}_{\lambda,\mathbb{R}}$, is real, but need not to be positive. In fact, if partition is not double, there is always a finite number of polynomials with negative squared norm, which can be easily identified using formula (\ref{Hsqnorm}). For example, setting $\lambda=(1)$ in \eqref{Hsqnorm}, we see that $\langle H_{(1),l},H_{(1),l}\rangle<0$ if and only if $l=0$. Grinevich and Novikov \cite{GN} pointed out a similar fact in a finite-gap case.
\end{remark}

We conclude this section by showing that the linear span $\mathcal U_\lambda$  is dense in $\mathcal{Q}_{\lambda,\mathbb{R}}$ in the sense that
$$
\langle p,H_{\lambda,l}\rangle=0,\ \forall l\in\mathbb{Z}_{\geq 0}\setminus\lbrace k_1,\ldots,k_n\rbrace\implies p\equiv 0.
$$
By Proposition \ref{Prop:Hbasis}, we can formulate the result as follows.

\begin{theorem}\label{Thm:Hdensity}
The subspace $\mathcal{Q}_\lambda$ is dense in $\mathcal{Q}_{\lambda,\mathbb{R}}$.
\end{theorem}

\begin{proof}
Suppose that $p\in\mathcal{Q}_{\lambda,\mathbb{R}}$ is such that
$$
\langle p,q\rangle=0,\ \ \ \forall q\in\mathcal{Q}_\lambda.
$$
Introducing the polynomials
$$
q_{\lambda,l}(z)=W^2_\lambda(z)H_l(z),\ \ \ l\in\mathbb{Z}_{\geq 0},
$$
which clearly belong to the subspace $\mathcal{Q}_\lambda$, we obtain
$$
0=\langle p,q_{\lambda,l}\rangle =\int_{i\xi+\mathbb{R}}p(z)\bar{H_l}(z)e^{-z^2}dz,\ \ \ \forall l\in\mathbb{Z}_{\geq 0}.
$$
Since the integrand is entire, we can take the limit $\xi\to 0$. Then expanding $p$ in terms of the classical Hermite polynomials $H_l$, it follows immediately from \eqref{clHorth} that $p\equiv 0$.
\end{proof}

\begin{remark}
If we assume that $\lambda$ is a double partition, then we recover orthogonality and completeness results from \cite{UGM-EHP} (see Propositions 5.7--5.8). Indeed, to recover the former it is enough to note that the weight function \eqref{w} is guaranteed to be non-singular on the real line, so that we can take the limit $\xi\to 0$ in \eqref{Hprod}; and the latter follows from the observation that we have $\mathcal{Q}_{\lambda,\mathbb{R}}=\mathbb{C}[z]$.
\end{remark}

\section{Exceptional Laurent orthogonal polynomials}
In this section we generalise our approach to the space of Laurent polynomials $\Lambda=\mathbb{C}[z,z^{-1}]$ using the trigonometric monodromy-free Schr\"odinger operators \cite{CFV-Huygen}, which play an important role in the theory of Huygens' principle \cite{BL}.

More specifically, we consider the Laurent polynomials $P_{\kappa,a;l}$, as defined in \eqref{excL}. Due to the results of Theorem \ref{Thm:Lorth} and Proposition \ref{Prop:Lext} we call $P_{\kappa,a;l}$, $l\in\mathbb Z$, exceptional Laurent orthogonal polynomials (ELOPs).

\subsection{The general case}
In this first subsection we allow any choice of complex parameters $a=(a_1,\ldots,a_n)$, $a_k\in\mathbb{C}\setminus\{0\}$.

We start from the elementary fact that the exponential functions
$$
e_l(x)=\exp(ilx),\ \ \ l\in\mathbb{Z},
$$
have the eigenfunction property
$$
\mathscr{L}e_l\equiv -\frac{d^2e_l}{dx^2}=l^2e_l,\ \ \ x\in\mathbb{C}/2\pi\mathbb{Z}.
$$
Note that instead of usual unit circle $\mathbb{R}/2\pi\mathbb{Z}$ we consider its complex version - cylinder $\mathbb{C}/2\pi\mathbb{Z}.$ It is natural from the trivial mondromy point of view, see \cite{CFV-Huygen}.

Sequences of Darboux transformations at the levels $0<k_n<k_{n-1}<\cdots<k_1$ are now parameterised by $n$ complex parameters $\theta=(\theta_1,\ldots,\theta_n)$, $\theta_k\in\mathbb{C}$. Specifically, introducing the functions
\begin{equation}
\label{Phik}
\phi_{k_j}(\theta_{j}, x)=2\cos(k_jx+\theta_{j}),\ \ \ j=1,\ldots,n,
\end{equation}
the resulting Schr\"odinger operator takes the form
\begin{equation}
\label{LlamL}
\mathscr{L}_\kappa=-\frac{d^2}{dx^2}-2\frac{d^2}{dx^2}\big(\log \textrm{Wr}(\phi_{k_1},\ldots,\phi_{k_n})\big),
\end{equation}
where $\kappa=\{k_1, \dots, k_n\}.$
Furthermore, letting $\mathscr{D}_\kappa$ act by
$$
\mathscr{D}_\kappa\phi=\frac{{\rm Wr}(\phi,\phi_{k_1},\ldots,\phi_{k_n})}{{\rm Wr}(\phi_{k_1},\ldots,\phi_{k_n})},
$$
the intertwining relation \eqref{inttw} holds true, and the functions
\begin{equation}
\label{Philaml}
\phi_{\kappa,\theta; l}=\frac{\textrm{Wr}(e_l,\phi_{k_1},\ldots,\phi_{k_n})}{\textrm{Wr}(\phi_{k_1},\ldots,\phi_{k_n})},\ \ \ l\in\mathbb{Z},
\end{equation}
have the eigenfunction property
$$
\mathscr{L}_\kappa\phi_{\kappa,l}=l^2\phi_{\kappa,l}.
$$
We note that at each level $k_j$, $j=1,\ldots, n$, the multiplicity is reduced from two to one. Indeed, by \eqref{Phik}--\eqref{Philaml} and linearity of the Wronskian, we have the relation
$$
\exp(i\theta_{j})\phi_{\kappa,k_j}(\theta_{j};x)+\exp(-i\theta_{j})\phi_{\kappa,-k_j}(\theta_{j};x)\equiv 0,\ \ \ j=1,\ldots,n.
$$

To establish the precise connection between the functions $\phi_{\kappa,l}$ and the ELOPs $P_{\kappa,a; l}$ given by (\ref{excL}), we change variable to
$$
z=\exp(ix)
$$
and fix the values of the parameters $a=(a_1,\ldots,a_n)$ according to
$$
a_k=\exp(i\theta_k)\in\mathbb{C}\setminus\{0\},\ \ \ k=1,\ldots,n.
$$
Then, it is readily seen that
$$
\phi_{\kappa,\theta;l}(\theta,x)=P_{\kappa,a; l}(z)\mathcal W_{\kappa,a} (z)^{-1},
$$
with $P_{\kappa,a; l}(z)$ given by (\ref{excL}) and
\begin{equation}
\label{hW}
\mathcal W_{\kappa,a}(z)=
\begin{vmatrix}
  	\Phi_{k_1}(a_1;z) 			&	\Phi_{k_2}(a_2;z)			&	\cdots 	&	\Phi_{k_n}(a_n;z)\\
  	D\Phi_{k_1}(a_1;z) 			&	D\Phi_{k_2}(a_2;z)			&	\cdots 	&	D\Phi_{k_n}(a_n;z)\\
 	\vdots  		& 	\vdots 		& 	\ddots 	& 	\vdots\\
  	D^{n-1}\Phi_{k_1}(a_1;z) 			&	D^{n-1}\Phi_{k_2}(a_2;z)			&	\cdots 	&	D^{n-1}\Phi_{k_n}(a_n;z)\\
\end{vmatrix},
\end{equation}
where $D=zd/dz$ and
$$
\Phi_k(a;z)=a z^k+a^{-1} z^{-k}.
$$
Furthermore, a direct computation reveals that $P_{\kappa,a; l}$ is an eigenfunction of the operator
$$
T_\kappa=-D^2+2\frac{D\mathcal W_{\kappa,a}}{\mathcal W_{\kappa,a}}D-\frac{D^2\mathcal W_{\kappa,a}}{\mathcal W_{\kappa,a}}
$$
with eigenvalue $l^2$.

\begin{example}
In the particular case $\kappa=\{1\}$ the corresponding Schr\"odinger operator is given by
\begin{align*}
\mathscr{L}_{\{1\}} &=-\frac{d^2}{dx^2}-2\frac{d^2}{dx^2}\big(\log(2\cos(x+\theta_1))\big)\\
&=-\frac{d^2}{dx^2}+\frac{2}{\cos^2(x+\theta_1)}.
\end{align*}
When expressed in terms of the variable $z$ and the parameter $a_1$, the first few exceptional Laurent polynomials $P_{\{1\},a; l}$ defined by (\ref{excL}) and the corresponding eigenfunctions $\Phi_{\{1\},a; l} = P_{\{1\},a; l}/\Phi_1, \, l \in \mathbb Z$ are given by
\renewcommand\arraystretch{2}
$$
\begin{array}{l r}
P_{\{1\},a; 0}=a_1z-a_1^{-1}z^{-1},\quad\quad   &\Phi_{\{1\},a; 0}=\frac{a_1z-a_1^{-1}z^{-1}}{a_1z+a_1^{-1}z^{-1}},\\
P_{\{1\},a; -1}=2a_1,\quad\quad&\Phi_{\{1\},a;-1}=\frac{2a_1}{a_1z+a_1^{-1}z^{-1}},\\
P_{\{1\},a;1}=-2a_1^{-1},\quad\quad&\Phi_{\{1\},a;1}=-\frac{2a_1^{-1}}{a_1z+a_1^{-1}z^{-1}},\\
P_{\{1\},a;-2}=a_1^{-1}z^{-3}+3a_1z^{-1},\quad\quad&\Phi_{\{1\},a;-2}=\frac{a_1^{-1}z^{-3}+3a_1z^{-1}}{a_1z+a_1^{-1}z^{-1}},\\
P_{\{1\},a;2}=a_1z^3+3a_1^{-1}z,\quad\quad&\Phi_{\{1\},a;2}=-\frac{a_1z^3+3a_1^{-1}z}{a_1z+a_1^{-1}z^{-1}}.
\end{array}
$$
From these explicit formulae, it is manifest that both $P_{\{1\},a;\pm 1}$ and $\Phi_{\{1\},a ;\pm 1}$ are linearly dependent and that each eigenfunction is singular at $z=\pm i/a_1$. For general $l\in\mathbb{Z}$, the latter fact can be easily seen from the definition of $P_{\{1\},a;l}$.
%
\end{example}

We note that, upon setting
$$
\mathscr{W}_m={\rm Wr}(\phi_{k_m},\ldots,\phi_{k_n}),
$$
the intertwining operator $\mathscr{D}_\kappa$ factorises according to \eqref{Dfact}--\eqref{Dm}. Just as in the Hermite case, it follows that each Schr\"odinger operator $\mathscr{L}_\kappa$ has trivial monodromy. Moreover, every monodromy-free trigonometric Schr\"odinger operator is of the form \eqref{LlamL}, see \cite{CFV-Huygen}.

Let $Z_\kappa$ be the set of zeros $z_i\in\mathbb{C}$ of the function $\mathcal W_{\kappa,a}(z)$ with multiplicities $m_i\in\mathbb{N}$ and $X_\kappa$ be the corresponding set consisting of $x_j$ such that $\exp(ix_j)=z_j, \,\, z_j\in Z_\kappa$ (we drop the dependence on $a$ in the notations for brevity in the rest of this section).

Introduce the subspace
$$
\mathscr{Q}_\kappa=\big\lbrace P\in\Lambda : \Phi(x):=(P/\mathcal W_\kappa)(\exp(ix))~\textrm{is quasi-invariant at all} \,\, x_j\in X_\kappa\big\rbrace.
$$
It follows from trivial monodromy property that
$$
\mathcal{U}_\kappa:=\langle P_{\kappa,l} : l\in\mathbb{Z}\rangle\subset\mathscr{Q}_\kappa.
$$
However, in contrast to Hermite case (see Proposition \ref{Prop:Hbasis}), the converse inclusion does not hold. Instead, we have the following result.

\begin{proposition}\label{Prop:Lcodim}
The codimension of $\mathcal{U}_\kappa$ in $\mathscr{Q}_\kappa$ is $n$.
\end{proposition}

\begin{proof}
From \eqref{excL}, we deduce that
$$
P_{\kappa,l}(z)=z^{l+|\kappa|}\det V(l,k_1,\ldots,k_n)\prod_{j=1}^nk_j+\mathrm{l.d.},
$$
where $$|\kappa|=\sum_{i=1}^n k_i,$$
 l.d.~stands for terms of lower degree and $V$ is the Vandermonde matrix
$$
V(\alpha_1,\ldots,\alpha_m)=
\left[\begin{matrix}
  	1 			&	1			&	\cdots 	&	1\\
  	\alpha_1 		&	\alpha_2		&	\cdots 	&	\alpha_m\\
 	\vdots  		& 	\vdots 		& 	\ddots 	& 	\vdots\\
  	\alpha_1^{m-1}	&	\alpha_2^{m-1}	&	\cdots 	&	\alpha_m^{m-1}\\
\end{matrix}\right].
$$ 
Since $\det V(l,k_1,\ldots,k_n)=0$ if and only if $l=k_1,\ldots,k_n$, it follows that the degree sequence
$$
I_\kappa^+=\{\deg P(z) : P\in \mathcal{U}_\kappa\}
$$
stabilises at $k_1+|\kappa|+1$ in the sense that $l\in I_\kappa^+$ for all $l\geq k_1+|\kappa|+1$. Applying the same line of reasoning to the Laurent polynomials $P_{\kappa,-l}(1/z)$, we find that the same statement holds true for
$$
I_\kappa^-=\{\deg P(z^{-1}) : P\in \mathcal{U}_\kappa\}.
$$
Among the ELOPs $P_{\kappa,l}$ with $|l|<k_1+|\kappa|+1$, a maximal set of linearly independent Laurent polynomials is given by
\begin{multline*}
l\in\{k_1,\ldots,k_n\}\cup\{0,\pm 1,\ldots,\pm (k_n-1)\}\cup\{\pm(k_n+1),\ldots,\pm(k_{n-1}-1)\}\\
\cup\cdots\cup\{\pm(k_2+1),\ldots,\pm(k_1-1)\}.
\end{multline*}
The cardinality of this index set equals
$$
n+2k_n-1+2(k_{n-1}-k_n-1)+\cdots+2(k_1-k_2-1)=2k_1-n+1.
$$
Observing that
$$
2k_1+2|\kappa|+1-(2k_1-n+1)=2|\kappa|+n,
$$
we conclude that the codimension of $\mathcal{U}_\kappa$ in $\Lambda$ is $2|\kappa|+n$.

On the other hand, counting quasi-invariance conditions, we find that the codimension of $\mathscr{Q}_\kappa$ in $\Lambda$ equals $2|\kappa|$ and so the assertion follows.
\end{proof}

\begin{remark} In contrast to the  case of usual polynomials there are several definitions of the degree of a Laurent polynomial, but none of them are convenient for our purposes. Let us define the {\it $L$-degree} $Ldeg P$ of a Laurent polynomial $P=\sum_{i=p}^{q} c_i z^i$ with $c_p\neq 0, c_q\neq 0$ as $q$ if $q>-p,$ and $p$ if $q<-p.$
If $q=-p$ the $L$-degree is not well-defined since it could be both $p$ and $q.$
Under these assumptions  $$\text{Ldeg}\, P_{\kappa,l}=|\kappa|+l, \, l \in \mathbb Z_+\setminus \kappa, \quad\quad \text{Ldeg}\, P_{\kappa,l}=-|\kappa|+l, \, -l \in \mathbb Z_+\setminus {\kappa},$$ otherwise it is not well-defined. Note that the polynomials $P_{\kappa,k_j}$ and $P_{\kappa,-k_j}$ with undefined $L$-degrees are linearly dependent.
\end{remark}

%
%

Next, we consider a particular complex bilinear form on $\mathscr{Q}_\kappa$, given by \eqref{complex_euclid} with $W=\mathcal W_\kappa$, and establish corresponding Laurent orthogonality relations. A related Fourier theory for more general algebro-geometric operators was studied by Grinevich and Novikov in \cite{GN}.

\begin{definition}
Let $\mu\in\mathbb{R}_{>0}$ be such that
\begin{equation}
\label{mures}
\mu\neq |z_i|,\ \ \ \forall z_i\in Z_\kappa.
\end{equation}
Then, we define a complex bilinear form $(\cdot,\cdot)$ on $\mathscr{Q}_\kappa$ by setting
\begin{equation}
(P,Q)=\frac{1}{2\pi i}\oint_{C_\mu}P(z)Q(z)\mathcal W_\kappa^{-2}\frac{dz}{z},\ \ \ P,Q\in\mathscr{Q}_\kappa,
\end{equation}
where
\begin{equation}
\label{Cmu}
C_\mu=\{z\in\mathbb{C} : |z|=\mu\}.
\end{equation}
\end{definition}

Substituting $z=\exp(ix)$ and following the line of reasoning used in the proof of Lemma \ref{Lemma:indep}, we readily find that the product is well-defined in the sense that it does not depend on the choice of $\mu$. More precisely, we have the following lemma.

\begin{lemma}
For any $P,Q\in\mathscr{Q}_\kappa$, the value of $(P,Q)$ is independent of $\mu\in\mathbb{R}_{>0}$ provided \eqref{mures} is satisfied.
\end{lemma}

We are now ready to state and prove the first of the main results in this section, which may be viewed as a natural analog of Theorem \ref{Thm:Horth}.

\begin{theorem}\label{Thm:Lorth}
The ELOPs $P_{\kappa,l}$ satisfy the Laurent orthogonality relation
$$
(P_{\kappa,j},P_{\kappa,l})=\delta_{j+l,0}\prod_{m=1}^n(l^2-k_m^2),\ \ \ j,l\in\mathbb{Z}.
$$
\end{theorem}

\begin{proof}
Just as in the proof of Theorem \ref{Thm:Horth}, we note that the assertion holds true for $n=0$, and proceed by induction on the length $n$ of $\kappa$. Letting $\hat{\kappa}=(k_2,\ldots, k_n)$, we have
$$
(P_{\kappa,j},P_{\kappa,l})=\frac{1}{2\pi}\int_0^{2\pi}(D_1\phi_{\hat{\kappa},j})(x)(D_1\phi_{\hat{\kappa},l})(x)dx.
$$
Making use of the factorisation
\begin{equation}
\label{fact}
D_1^*D_1=\mathscr{L}_{\hat{\kappa}}-k_1^2,
\end{equation}
with
$$
D_1^*=-\frac{d}{dx}-\frac{d}{dx}\left(\log\frac{\mathscr{W}_1}{\mathscr{W}_2}\right)
$$
the (formal) adjoint of $D_1$, we deduce
$$
(P_{\kappa,j},P_{\kappa,l})=(l^2-k_1^2)(P_{\hat{\kappa},j},P_{\hat{\kappa},l}),
$$
which completes the induction step.
\end{proof}

\begin{remark}
Having started from an eigenvalue problem with doubly degenerate eigenvalues, we have that $(P_{\kappa,l},P_{\kappa,-l})=0$ for some of the ELOPs $P_{\kappa,l}$. More specifically, it is evident from the theorem that this is the case if and only if $l=\pm k_m$, $m=1,\ldots,n$.
\end{remark}

Expanding on the result of Proposition \ref{Prop:Lcodim}, we proceed to establish the precise relationship between $\mathcal{U}_\kappa$ and $\mathscr{Q}_\kappa$. We begin with a general definition.

Let $V$ be a vector space over $\mathbb C.$ Then $V$ is called {\it complex Euclidean space} if it is equipped with a non-degenerate bilinear form $B:V\otimes V\to\mathbb{C}.$

\begin{definition}\label{Def:mext}
Let $W\subset V$ be a subspace of complex Euclidean space $V.$ We say that $V$ is a minimal complex Euclidean extension of $W$ if
$$
\dim\big(\ker B\arrowvert_W\big)=\mathrm{codim}_V W.
$$
\end{definition}

For any linear space $W$ and bilinear form $B$ with non-trivial kernel
$$
K:=\ker B,
$$
it is readily verified that there is a unique (up to isomorphisms) minimal complex Euclidean extension $V\supset W$. Letting $K^*$ denote the dual space of $K$, it can be realised as follows:
$$
V=K\oplus K^*\oplus W/K,
$$
with the extension of $B$ determined by
$$
(k_1+\hat{k}_1+w_1,k_2+\hat{k}_2+w_2)\mapsto \hat{k}_2(k_1)+\hat{k}_1(k_2)+B(w_1,w_2),
$$
where $k_1,k_2\in K$, $\hat{k}_1,\hat{k}_2\in K^*$ and $w_1,w_2\in W$. Moreover, for each basis $k_1,\ldots,k_n\in K$, there is a unique basis $\hat{k}_1,\ldots,\hat{k}_n\in K^*$ such that $(k_j,\hat{k}_l)=\delta_{jl}$.

\begin{example}
Suppose that $B\arrowvert_W=0$, so that each vector $w\in W$ is isotropic. Then we have
$$
V\cong W\oplus W^*,
$$
with
$$
B(w_1+\hat{w}_1,w_2+\hat{w}_2)=\hat{w}_2(w_1)+\hat{w}_1(w_2),\ \ \ w_1,w_2\in W,\ \ \hat{w}_1,\hat{w}_2\in W^*.
$$
\end{example}

As demonstrated by the following proposition, the inclusion $\mathcal{U}_\kappa\subset\mathscr{Q}_\kappa$ provides a concrete example of a minimal complex Euclidean extension in the sense of Definition \ref{Def:mext}.

\begin{proposition}\label{Prop:Lext}
$\mathscr{Q}_\kappa$ is the minimal complex Euclidean extension of $\mathcal{U}_\kappa$.
\end{proposition}

\begin{proof}
From Theorem \ref{Thm:Lorth} we infer that
$$
\ker (\cdot,\cdot)\arrowvert_{\mathcal{U}_\kappa}=\langle P_{\kappa,k_j} : j=1,\ldots,n\rangle.
$$
(Note the linear relations $a_jP_{\kappa,k_j}+a_j^{-1}P_{\kappa,-k_j}=0$.) Since ELOPs $P_{\kappa,l}$ corresponding to different values of $l^2$, and hence different eigenvalues, are linearly independent, it follows that
$$
\dim\big(\ker (\cdot,\cdot)\arrowvert_{\mathcal{U}_\kappa}\big)=n.
$$
Recalling Proposition \ref{Prop:Lcodim}, we see that it remains only to verify that $(\cdot,\cdot)$ is non-degenerate on $\mathscr{Q}_\kappa$. Observing that
$$
\mathcal W_\kappa^2(z)z^j\in\mathscr{Q}_\kappa,\ \ \ \forall j\in\mathbb{Z},
$$
this follows, e.g., from the computation
\begin{align*}
\big(P_{\kappa,l},\mathcal W_\kappa^2(z)z^{-l-|\kappa|}\big) &= \frac{1}{2\pi i}\oint_{C_\mu}P_{\kappa,l}(z)z^{-l-|\kappa|}\frac{dz}{z}\\
&= \det V(l,k_1,\ldots,k_n)\prod_{j=1}^nk_j,
\end{align*}
which is non-zero as long as $l\neq\pm k_j$, cf.~the proof of Proposition \ref{Prop:Lcodim}.
\end{proof}

\subsection{The Hermitian case}

In the case when all $\theta_k$ are real or, equivalently, when parameters $a=(a_1,\ldots,a_n)$ satisfy
\begin{equation}
\label{akres}
|a_k|=1,\ \ \ k=1,\ldots,n,
\end{equation}
we can introduce the Hermitian structure as follows.

Note that in this case the weight function $w(z)=\mathcal W_\kappa(z)^{-2}$ is invariant under the antilinear involution
\begin{equation}
\label{inv}
P^\dagger(z):=\overline{P(1/\bar{z})},\ \ \ P\in\Lambda,
\end{equation}
which will play much the same role as the Schwartz conjugate did in the Hermite case. In fact, observing that $(DP)^\dagger=-DP^\dagger$ and that $\Phi_k^\dagger=\Phi_k$, we can deduce from \eqref{hW} that
\begin{equation}
\label{Winv}
\mathcal W_\kappa^\dagger(z)=(-1)^{n(n-1)/2}\mathcal W_\kappa(z),\ \ \ \kappa=\{k_1,\ldots, k_n\}.
\end{equation}
In addition, the zero set $Z_\kappa$ becomes invariant under the involution $z\to 1/\bar{z}$, i.e.
$$
z_i\in Z_\kappa\implies 1/\bar{z_i}\in Z_\kappa,
$$
and, since $z=1/\bar{z}$ whenever $|z|=1$, we have that
\begin{equation}
\label{Wprod}
\mathcal W_\kappa(z)\mathcal W_\kappa^\dagger(z)=|\mathcal W_\kappa|^2,\ \ \ |z|=1.
\end{equation}

Letting $Z_\kappa^C=\{z_i\in Z_\kappa : |z|=1\}$ and $X_\kappa^{\mathbb R}=\{x_j: \exp(ix_j)=z_j, \,\, z_j\in Z_\kappa^C\} \subset \mathbb R$, we introduce the following subspace of quasi-invariant Laurent polynomials:
$$
\mathscr{Q}_{\kappa,C}
=\left\lbrace P\in\Lambda : \Phi(x):=(P/\mathcal W_\kappa)(\exp(ix))~\textrm{is quasi-invariant at all}~\,\, x_j\in X_\kappa^{\mathbb R}\right\rbrace.
$$
From \eqref{Winv}, it is straightforward to infer that
$$
\mathscr{Q}_\kappa^\dagger=\mathscr{Q}_\kappa,\ \ \ \mathscr{Q}_{\kappa,C}^\dagger=\mathscr{Q}_{\kappa,C},
$$
which allows us to define a natural sesquilinear product on $\mathscr{Q}_{\kappa,C}$.

\begin{definition}\label{Def:Lprod}
Assuming that \eqref{akres} holds true, we introduce
$$
\nu=\min_{\substack{z_i\in Z_\kappa\\ |z_i|>1}}|z_i|,
$$
and let $\mu\in\mathbb{R}_{>0}$ be such that
\begin{equation}
\label{mures2}
1<\max(\mu,1/\mu)<\nu.
\end{equation}
Then, we define a sesquilinear product $\langle\cdot,\cdot\rangle_L$ on $\mathscr{Q}_{\kappa,C}$ by setting
\begin{equation}
\label{Lprod}
\langle P,Q\rangle_L=\frac{1}{2\pi i}\oint_{C_\mu}P(z)Q^\dagger(z)\big(\mathcal W_\kappa(z)\mathcal W_\kappa^\dagger(z)\big)^{-1}\frac{dz}{z},\ \ \ P,Q\in\mathscr{Q}_{\kappa,C}.
\end{equation}
\end{definition}

Again, the product does not depend on the specific choice of $\mu$.

\begin{lemma}\label{Lemma:indep2}
For any $P,Q\in\mathscr{Q}_{\kappa,C}$, the value of $(P,Q)_L$ is independent of $\mu\in\mathbb{R}_{>0}$ provided \eqref{mures2} is satisfied.
\end{lemma}

By adapting the proof of Proposition \ref{Prop:Hprod}, we can use the lemma to show that Definition \ref{Def:Lprod} yields a Hermitian product.

\begin{proposition}
The sesquilinear product $\langle\cdot,\cdot\rangle_L$ is Hermitian:
$$
\langle P,Q\rangle_L=\overline{\langle Q,P\rangle_L},\ \ \ \forall P,Q\in\mathscr{Q}_{\kappa,C}.
$$
\end{proposition}

\begin{proof}
Using the notation
$$
w(z)=1\big/\mathcal W_\kappa(z)\mathcal W_\kappa^\dagger(z)
$$
and using a subscript to indicate the choice of $\mu$ in \eqref{Lprod}, we deduce the following equalities:
\begin{align*}
\langle P,Q\rangle_{L,\mu} &=\frac{1}{2\pi}\int_0^{2\pi}P(\mu e^{i\varphi})\bar{Q}(\mu^{-1}e^{i\varphi})w(\mu e^{i\varphi})d\varphi\\
&=\frac{1}{2\pi}\overline{\int_0^{2\pi}\bar{P}(\mu e^{-i\varphi})Q(\mu^{-1}e^{i\varphi})w(\mu^{-1}e^{i\varphi})d\varphi}\\
&=\overline{\langle Q,P\rangle_{L,\mu^{-1}}},
\end{align*}
and so hermiticity follows from Lemma \ref{Lemma:indep2}.
\end{proof}

Moreover, the proof of Theorem \ref{Thm:Lorth} is readily adapted to yield the following orthogonality result.

\begin{theorem}
Assuming that \eqref{akres} holds true, the ELOPs $P_{\kappa,l}$ satisfy the orthogonality relation
\begin{equation}
\label{Lorth}
\langle P_{\kappa,j},P_{\kappa,l}\rangle_L=\delta_{jl}\prod_{m=1}^n(k_m^2-l^2),\ \ \ j,l\in\mathbb{Z}.
\end{equation}
\end{theorem}

\begin{proof}
Taking $z=\exp(ix)$ in the integral in \eqref{Lprod} and observing that (cf.~\eqref{Winv})
$$
\overline{\mathscr{W}_m}(-x)=(-1)^{(n-m)(n-m+1)/2}\mathscr{W}_m(x),
$$
we establish the equalities
\begin{align*}
\langle P_{\kappa,j},P_{\kappa,l}\rangle_L &=\frac{1}{2\pi}\int_0^{2\pi}(D_1\phi_{\hat{\kappa},j})(x)\overline{(D_1\phi_{\hat{\kappa},l})}(-x)dx\\
&=-\frac{1}{2\pi}\int_0^{2\pi}\phi_{\hat{\kappa},j}(x)\overline{(D_1^*D_1\phi_{\hat{\kappa},l})}(-x)dx.
\end{align*}
Appealing to the factorisation \eqref{fact}, we thus obtain the relation
$$
\langle P_{\kappa,j},P_{\kappa,l}\rangle_L=(k_1^2-l^2)\langle P_{\hat{\kappa},j},P_{\hat{\kappa},l}\rangle_L,
$$
and the assertion follows by induction on $n$.
\end{proof}

After replacing the bilinear form $B$ by a Hermitian sesquilinear form $h$, Definition \ref{Def:mext} as well as the succeeding discussion applies with minor changes also in the present situation. Specifically, we say that $V$ is a {\it minimal  Hermitian extension} of $W$ if
$$
\dim\big(\ker h\arrowvert_W\big)=\mathrm{codim}_V W.
$$
Then, we have the following analog of Theorem \ref{Thm:Hdensity}.

\begin{theorem}\label{Thm:Ldensity}
The subspace $\mathscr{Q}_\kappa$, which is the minimal Hermitian extension of $\mathcal{U}_\kappa$, is dense in $\mathscr{Q}_{\kappa,C}$.
\end{theorem}

\begin{proof}
Suppose that $P\in\mathscr{Q}_{\kappa,C}$ is such that
$$
\langle P,Q\rangle_L=0,\ \ \ \forall Q\in\mathscr{Q}_\kappa.
$$
Since the Laurent polynomials
$$
Q_{\kappa,l}=\mathcal W_\kappa(z)\mathcal W_\kappa^\dagger(z)z^l,\ \ \ l\in\mathbb{Z},
$$
clearly are contained in $\mathscr{Q}_\kappa$, we have that
$$
0=\langle P,Q_{\kappa,l}\rangle_L=\frac{1}{2\pi i}\oint_{C_\mu}P(z)z^l\frac{dz}{z},\ \ \ \forall l\in\mathbb{Z}.
$$
Taking the limit $\mu\to 1$ and using the property that
$$
\frac{1}{2\pi i}\oint_{|z|=1}z^kz^l\frac{dz}{z}=\delta_{k+l,0},\ \ \ k,l\in\mathbb{Z}
$$
we conclude that $P\equiv 0$.
\end{proof}


\begin{remark}
It is known from the soliton theory that for every non-empty set $\kappa$ and any choice of real $\theta_k$ the corresponding potential always has singularities on the real line. This means that in the Laurent case we do not have non-trivial regular examples (unlike Hermite case with double partitions).
\end{remark}

\section{Concluding remarks}

We have discussed two complex versions of the exceptional orthogonal polynomials, related to two classes of monodromy-free Schr\"odinger operators. We would like to emphasize two novelties compared to the original approach of G\'omez-Ullate et al \cite{UKM-boch,UGM-EHP}.

First, in order to define the inner product in general we have to reduce the space of polynomials to the subspace of quasi-invariants, which has a finite codimension. The only exception is the Hermite case with double partitions considered in \cite{UGM-EHP}. 

Second, in the Laurent case the space of quasi-invariants is not generated by the corresponding exceptional Laurent polynomials, so we need to consider the minimal complex Euclidean extension.

In the rational case with sextic growth at infinity there are some partial results \cite{GV}, which lead to finite sets of orthogonal polynomials of the same degree. It would be interesting to analyse this situation in a view of recent very interesting paper by Felder and Willwacher \cite{FW}.

It would be interesting also to see what happens with exceptional orthogonal polynomials in the multidimensional case. One can use the monodromy-free generalised Calogero-Moser operators, playing an important role in the theory of Huygens principle \cite{CFV-Huygen}.
We plan to address this elsewhere soon.

\section{Acknowledgements.}

APV is grateful to Oleg Chalykh and Leon Takhtajan for very useful and encouraging discussions.

\end{document}